\documentclass[aps,twocolumn,nofootinbib]{revtex4}
\usepackage{amsfonts}
\usepackage{amsmath}
\usepackage{amssymb}
\usepackage{graphicx}
\newtheorem{theorem}{Theorem}

\newtheorem{definition}[theorem]{Definition}

\newtheorem{lemma}[theorem]{Lemma}

\newenvironment{proof}[1][Proof]{\noindent\textbf{#1.} }{\ \rule{0.5em}{0.5em}}

\begin{document}
\preprint{ }
\title{Einstein, incompleteness, and the epistemic view of quantum states}
\author{Nicholas Harrigan}
\affiliation{QOLS, Blackett Laboratory, Imperial College London,
Prince Consort Road, London SW7 2BW, United Kingdom}
\author{Robert W. Spekkens}
\affiliation{Department of Applied Mathematics and Theoretical Physics, University of
Cambridge, Cambridge CB3 0WA, United Kingdom}
\date{June 15, 2007}

\begin{abstract}
Does the quantum state represent reality or our knowledge of
reality? In making this distinction precise, we are led to a novel
classification of hidden variable models of quantum theory. Indeed,
representatives of each class can be found among existing
constructions for two-dimensional Hilbert spaces. Our approach also
provides a fruitful new perspective on arguments for the nonlocality
and incompleteness of quantum theory. Specifically, we show that for
models wherein the quantum state has the status of something real,
the failure of locality can be established through an argument
considerably more straightforward than Bell's theorem. The
historical significance of this result becomes evident when one
recognizes that the same reasoning is present in Einstein's
preferred argument for incompleteness, which dates back to 1935.
This fact suggests that Einstein was seeking not just \textit{any}
completion of quantum theory, but one wherein quantum states are
solely representative of our knowledge. Our hypothesis is supported
by an analysis of Einstein's attempts to clarify his views on
quantum theory and the circumstance of his otherwise puzzling
abandonment of an even simpler argument for incompleteness from
1927.
\end{abstract}


\maketitle

\section{Introduction \label{SEC:intro}}

We explore a distinction among hidden variable models of quantum
theory that has hitherto not been sufficiently emphasized, namely,
whether the quantum state is considered to be ontic or epistemic. We
call a hidden variable model \emph{$\psi$-ontic} if every complete
physical state or \emph{ontic state} \cite{Spekkens_con} in the
theory is consistent with only one pure quantum state; we call it
\emph{$\psi$-epistemic} if there exist ontic states that are
consistent with more than one pure quantum state. In $\psi$-ontic
models, distinct quantum states correspond to disjoint probability
distributions over the space of ontic states, whereas in
$\psi$-epistemic models, there exist distinct quantum states that
correspond to overlapping probability distributions. Only in the
latter case can the quantum state be considered to be truly
epistemic, that is, a representation of an observer's knowledge of
reality rather than reality itself. (This distinction will be
explained in detail further on.)

It is interesting to note that, to the authors' knowledge, all
mathematically explicit hidden variable models proposed to date are
$\psi$-ontic (with the exception of a proposal by Kochen and Specker
\cite{Ks} that only works for a two-dimensional Hilbert space and
which we will discuss further on).\footnote{Subtleties pertaining to
Nelson's mechanics and unconventional takes on the deBroglie-Bohm
interpretation will also be discussed in due course.} The study of
$\psi$-epistemic hidden variable models is the path less traveled in
the hidden variable research program. This is unfortunate given that
recent work has shown how useful the assumption of hidden variables
can be for explaining a variety of quantum phenomena $\emph{if}$ one
adopts a $\psi$-epistemic approach
\cite{Hardydisentangling,toy_theory,BRSLiouville,tr_model}.

It will be useful for us to contrast hidden variable models with the
interpretation that takes the quantum state alone to be a complete
description of reality. We call the latter the
\emph{$\psi$-complete} view, although it is sometimes referred to as
the \emph{orthodox} interpretation\footnote{Note that while Bohr
argued for the completeness of the quantum state, he did so within
the context of an instrumentalist rather than a realist approach and
consequently his view is not the one that we are interested in
examining here. Despite this, the realist $\psi$-complete view we
have in mind does approximate well the views of many researchers
today who identify themselves as proponents of the Copenhagen
interpretation.}.

Arguments against the $\psi$-complete view and in favor of hidden
variables have a long history. Among the most famous are those that
were provided by Einstein. Although he did not use the term `hidden
variable interpretation', it is generally agreed that such an
interpretation captures his approach. Indeed, Einstein had attempted
to construct a hidden variable model of his own (although ultimately
he did not publish this work)
\cite{Howard_eghost,bac_valentini_Epilotwave}. One of the questions
we address in this article is whether Einstein favored either of the
two sorts of hidden variable theories we have outlined above:
$\psi$-ontic or $\psi$-epistemic. Experts in the quantum foundations
community have long recognized that Einstein had already shown a
failure of locality for the $\psi$-complete view with a very simple
argument at the Solvay conference in 1927 \cite{bac_valentini}. It
is also well-known in such circles that a slightly more complicated
argument given in 1935 --- one appearing in his correspondence with
Schr\"{o}dinger, not the Einstein-Podolsky-Rosen paper --- provided
yet another way to see that locality was ruled out for the
$\psi$-complete view\footnote{Borrowing a phrase from Asher Peres
\cite{Peres_noclone}, these facts are ``well known to those who know
things well''.} \cite{Howard_einst_short,FineEcritique}.\ What is
not typically recognized, and which we show explicitly here, is that
the latter argument was actually strong enough to also rule out
locality for $\psi$-ontic hidden variable theories. In other words,
Einstein showed that not only is locality inconsistent with $\psi$
being a complete description of reality, it is also inconsistent
with $\psi$ being ontic, that is, inconsistent with the notion that
$\psi$ represents reality even in an incomplete sense. Einstein thus
provided an argument for the epistemic character of $\psi$ based on
locality.

Fuchs has previously argued in favor of this conclusion. In his
words, ``[Einstein] was the first person to say in absolutely
unambiguous terms why the quantum state should be viewed as
information [...]. His argument was simply that a quantum-state
assignment for a system can be forced to go one way or the other by
interacting with a part of the world that should have no causal
connection with the system of interest.'' \cite{Fuchs_Eepistemic}.
One of the main goals of the present article is to lend further
support to this thesis by clarifying the relevant concepts and by
undertaking a more detailed exploration of Einstein's writings. We
also investigate the implications of our analysis for the history of
incompleteness and nonlocality arguments in quantum theory.

In particular, our analysis helps to shed light on an interesting
puzzle regarding the evolution of Einstein's arguments for
incompleteness.

The argument Einstein gave at the 1927 Solvay conference requires
only a single measurement to be performed, whereas from 1935 onwards
he adopted an argument requiring a measurement to be chosen from two
possibilities. Why did Einstein complicate the argument in this way?
Indeed, as has been noted by many authors, this complication was
actually detrimental to the effectiveness of the argument, given
that most of the criticisms directed against the two-measurement
form of the argument (Bohr's included) focus upon his use of
counterfactual reasoning, an avenue that is not available in the
1927 version \cite{Hardy1995,Redhead,Fine,Maudlin,norsenboxes}.

The notion that Einstein introduced this two-measurement
complication in order to simultaneously beat the uncertainty
principle, though plausible, is not supported by textual evidence.
Although the Einstein-Podolsky Rosen (EPR) paper \emph{does} take
aim at the uncertainty principle, it was written by Podolsky and, by
Einstein's own admission, did not provide an accurate synopsis of
his (Einstein's) views. This has been emphasized by Fine
\cite{FineEcritique} and Howard \cite{Howard_einstnotepr}. In the
versions of the argument that were authored by Einstein, such as
those appearing in his correspondence with Schr\"{o}dinger, the
uncertainty principle is explicitly de-emphasized. Moreover, to the
authors' knowledge, whenever Einstein summarizes his views on
incompleteness in publications or in his correspondence after 1935,
it is the argument appearing in his correspondence with
Schr\"{o}dinger, rather than the EPR argument, to which he appeals.

We suggest a different answer to the puzzle. Einstein consistently
used his more complicated 1935 argument in favor of his simpler 1927
one because the extra complication bought a stronger conclusion,
namely, that the quantum state is not just incomplete, but
\textit{epistemic}. We suggest that Einstein implicitly recognized
this fact, even though he failed to emphasize it adequately.

Finally, our results demonstrate that one doesn't need the
\textquotedblleft big guns\textquotedblright\ of Bell's theorem
\cite{Bell_locality} to rule out locality for any theories in which
$\psi$ is given ontic status; more straightforward arguments
suffice. Bell's argument is only necessary to rule out locality for
$\psi$-epistemic hidden variable theories. It is therefore
surprising that the latter sort of hidden variable theory, despite
being the most difficult to prove inconsistent with locality and
despite being the last, historically, to have been subject to such a
proof, appears to have somehow attracted the least attention, with
Einstein a notable but lonely exception to the rule.

\section{The distinction between $\psi$-ontic and $\psi$-epistemic ontological
models}

\subsection{What is an ontological model? \label{SEC:om_intro}}

We begin by defining some critical notions. First is that of an
ontological model of a theory. Our definition will require that the
theory be formulated operationally, which is to say that the
primitives of description are simply preparation and measurement
procedures -- lists of instructions of what to do in the lab. The
goal of an operational formulation of a theory is simply to
prescribe the probabilities of the outcomes of different
measurements given different preparation procedures, that is, the
probability $p(k|M,P)$ of obtaining outcome $k$ in measurement $M$
given preparation $P.$ For instance, in an operational formulation
of quantum theory, every preparation $P$ is associated with a
density operator $\rho$ on Hilbert space, and every measurement $M$
is associated with a positive operator valued measure (POVM)
$\{E_{k}\}$. \ (In special cases, these may be associated with
vectors in Hilbert space and Hermitian operators respectively.) The
probability of obtaining outcome $k$ is given by the generalized
Born rule, $p(k|M,P)=\mathrm{Tr}(\rho E_{k}).$

In an ontological model of an operational theory, the primitives of
description are the properties of microscopic systems. \ A
preparation procedure is assumed to prepare a system with certain
properties and a measurement procedure is assumed to reveal
something about those properties. \ A complete specification of the
properties of a system is referred to as the \emph{ontic state} of
that system, and is denoted by $\lambda$. The ontic state space is
denoted by $\Lambda$. It is presumed that an observer who knows the
preparation $P$ may nonetheless have incomplete knowledge of
$\lambda.$ \ In other words, the observer may assign a non-sharp
probability distribution $p(\lambda|P)$ over $\Lambda$ when the
preparation is known to be $P.$ Similarly, the model may be such
that the ontic state $\lambda$ determines only the probability
$p(k|\lambda,M)$ of different outcomes $k$ for the measurement $M.$
\ We shall refer to $p(\lambda|P)$ as an \emph{epistemic state},
because it characterizes the observer's knowledge of the system. \
We shall refer to $p(k|\lambda,M),$ considered as a function of
$\lambda,$ as an \emph{indicator function}. For the ontological
model to reproduce the predictions of the operational theory, it
must reproduce the probability of $k$ given $M$ and $P$ through the
formula $\int\mathrm{d}\lambda
{p}(k|M,\lambda)p(\lambda|P)=p(k|M,P).$

An ontological model of quantum theory is therefore defined as
follows.
\begin{definition}
An ontological model of operational quantum theory posits an ontic
state space $\Lambda$ and prescribes a probability distribution over
$\Lambda$ for every preparation procedure $P$, denoted
$p(\lambda|P)$, and a probability distribution over the different
outcomes $k$ of a measurement $M$ for every ontic state
$\lambda\in\Lambda,$ denoted\ $p(k|\lambda,M).$ Finally, for all $P$
and $M,$ it must satisfy,
\begin{equation}
\int\mathrm{d}\lambda{p}(k|M,\lambda)p(\lambda|P)=\mathrm{tr}\left(  \rho
{E}_{k}\right), \label{ont_mod_qm_stats}%
\end{equation}
where $\rho$ is the density operator associated with $P$ and $E_{k}$ is the
POVM element associated with outcome $k$ of $M$.
\end{definition}
The structure of the posited $\Lambda$ encodes the kind of reality
envisaged by the model, while $p(\lambda|P)$ and $p(k|M,\lambda)$
specify what can be known and inferred by observers. Note that we
refer to preparation and measurement procedures rather than quantum
states and POVMs because we wish to allow for the possibility of
contextual\footnote{In a preparation (measurement) contextual
ontological model, different preparation (measurement) procedures
corresponding to the same density operator (POVM) may be assigned
different epistemic states (indicator functions) by the ontological
model.} ontological models \cite{Spekkens_con}.

Note that although the ontological model framework proposed here is
very general, there could exist realist interpretations of quantum
theory that are not suited to it. However, the vast majority of
models analyzed so far seem compatible with it (or a simple
extension to be given in \cite{deficiency}).

\subsection{Classifying ontological models of quantum theory:
heuristics\label{SEC:classes_first}}

An important feature of an ontological model is how it takes the
quantum states describing a system to be related to the ontic states
of that system. The simplest possibility is a one-to-one
relation.\footnote{Note that it is because of such models, wherein
nothing is hidden to one who knows the quantum state, that we adopt
the term ``ontological model'' as opposed to ``hidden variable
model''. Some authors might prefer to use the latter term on the
grounds that a $\psi$-complete model is simply a trivial instance of
a hidden variable model, but we feel that such a terminology would
be confusing.} A schematic of such a model is presented in part~(a)
of Fig.~\ref{FIG:classes}, where we have represented the set of all
quantum states by a one-dimensional ontic state space $\Lambda$
labeled by $\psi.$ We refer to such models as $\psi$\emph{-complete}
because a pure quantum state provides a complete description of
reality. Many might consider this to be the `orthodox'
interpretation.

\begin{figure}[t]
\includegraphics[scale=0.4]{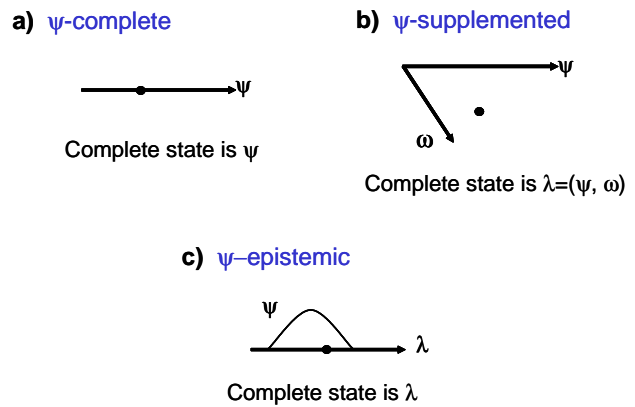}\caption{Schematic view of the ontic state
space for \textbf{(a)} $\psi$-complete models, \textbf{(b)}
$\psi$-supplemented
models and \textbf{(c)} $\psi$-epistemic models.}%
\label{FIG:classes}%
\end{figure}

Of course, the ontological model framework also allows for the
possibility that a complete description of reality may require
\textit{supplementing} the quantum state with additional variables.
Such variables are commonly referred to as `hidden' because their
value is typically presumed to be unknown to someone who knows the
identity of the quantum state. In such models, knowledge of $\psi$
alone provides only an \textit{incomplete} description of reality.

The ontic state space for such a model is schematized in part (b) of
Fig.~\ref{FIG:classes}. Although there may be an arbitrary number of
hidden variables, we indicate only a single hidden variable $\omega$
in our diagram, represented by an additional axis in the ontic state
space $\Lambda.$ Specification of the complete ontic configuration
of a system (a point $\lambda\in\Lambda$) now requires specifying
both $\psi$ and the hidden variable $\omega$. We refer to models
wherein $\psi$ must be supplemented by hidden variables as $\psi
$\emph{-supplemented}. Almost all ontological models of quantum
mechanics constructed to date have fallen into this class. For
example, in the conventional view of the deBroglie-Bohm
interpretation \cite{bohm,bohmsurvey}, the complete ontic state is
given by $\psi$ together with (that is, supplemented by) the
positions of all particles. The ontic nature of\ $\psi$ in the
deBroglie-Bohm interpretation is clear from the fact that it plays
the role of a pilot wave, so that distinct $\psi$s describe
physically distinct universes. Bell's `beable' interpretations
\cite{beables} and modal interpretations of quantum mechanics
\cite{modalrefs1,modalrefs2,modalrefs3,modalrefs4} also take $\psi$
to be a sort of pilot wave and thus constitute $\psi$-supplemented
models \footnote{Note that another way in which to express how
$\psi$-complete and $\psi$-supplemented models differ from
$\psi$-epistemic models is that only in the former is $\psi$ itself
a beable \cite{beables}}. As another example, Belifante's survey of
hidden variable theories \cite{Belifante} considers only
$\psi$-supplemented models.

There is a different way in which $\psi$ could be an incomplete
description of reality: it could represent a state of incomplete
knowledge about reality. In other words, it could be that $\psi$ is
not a variable in the ontic state space at all, but rather encodes a
probability distribution over the ontic state space. In this case
also, specifying $\psi$ does not completely specify the ontic state,
and so it is apt to say that $\psi$ provides an incomplete
description. In such a model, a variation of $\psi$ does not
represent a variation in any physical degrees of freedom, but
instead a variation in the space of possible ways of knowing about
some underlying physical degrees of freedom. This is illustrated
schematically in part (c) of Fig.~\ref{FIG:classes}. We refer to
such models as $\psi$-epistemic.\footnote{There is, however, a
subtlety in ensuring that a probability distribution associated with
$\psi$ is truly epistemic; we address this issue shortly.}

\subsection{Classifying ontological models of quantum theory: a more rigorous
approach\label{SEC:classes_second}}

It will be convenient for our purposes to provide precise
definitions of $\psi$-complete, $\psi$-supplemented, and
$\psi$-epistemic models in terms of the epistemic states that are
associated with different $\psi$. \ In other words, for each model,
we enquire about the probability distribution over the ontic state
space that is assigned by an observer who knows that the preparation
procedure is associated with the quantum state $\psi.$ Despite
appearances, this does not involve any loss of generality. \ For
instance, although it might appear that $\psi$-complete models can
only be defined by their ontological claims, namely, that pure
quantum states are associated one-to-one with ontic states, such
claims can always be re-phrased as epistemic claims, in this case,
that knowing the quantum state to be $\psi$ implies having a state
of complete knowledge about the ontic state.

We now provide precise definitions of two distinctions among
ontological models from which one can extract the three categories
introduced in Sec.~\ref{SEC:classes_first}. The first distinction is
between models that are $\psi$-complete and those that are not.

\begin{definition}
\strut An ontological model is \textbf{$\psi$-complete} \strut if
the ontic state space $\Lambda$ is isomorphic to the projective
Hilbert space $\mathcal{PH}$ (the space of rays of Hilbert space)
and if every preparation procedure $P_{\psi}$ associated in quantum
theory with a given ray $\psi$ is associated in the ontological
model with a Dirac delta function centered at the ontic state
$\lambda_{\psi}$ that is isomorphic to $\psi$, ${p}(\lambda
|P_{\psi})=\delta(\lambda-\lambda_{\psi}).\footnote{The Dirac delta
function on $\Lambda$ is defined by
$\int_{\Lambda}\delta(\lambda-\lambda_{\psi
})f(\lambda)\mathrm{d}\lambda=f(\lambda_{\psi}).$}$
\label{DEF:psi_complete}
\end{definition}

Hence, in such models, the only feature of the preparation that is
important is the pure quantum state to which it is associated.
Epistemic states for a pair of preparations associated with distinct
quantum states are illustrated schematically in part (a) of
Fig.~\ref{FIG:classes2}.\footnote{In the case of a mixture of pure
states, one uses the associated mixture of epistemic states. For
instance, if the preparation is of a pure state $\psi_{i}$ with
probability $w_{i},$ then the epistemic state is
$\sum_{i}w_{i}p(\lambda|\psi _{i})$. Note, however, that it is not
at all clear how to deal in a $\psi$-complete model with
\emph{improper} mixtures, that is, mixed density operators that
arise as the reduced density operator of an entangled state. \ This
fact is often used to criticize such models.}

\begin{definition}
\strut If an ontological model is not $\psi$-complete, then it is said to be
\textbf{$\psi$-incomplete. }\label{DEF:psi_incomplete}
\end{definition}

Identifying a model as $\psi $-incomplete does not specify how such
a failure is actually manifested. It might be that $\Lambda$ is
parameterized by $\psi$ \textit{and} by supplementary variables, or
it could alternatively be that the quantum state does not
parameterize the ontic states of the model at all. In order to be
able to distinguish these two possible manifestations of $\psi
$-incompleteness, we introduce a second dichotomic classification of
ontological models.

\begin{definition}
An ontological model is \textbf{$\psi$-ontic} if for any pair of
preparation procedures, $P_{\psi}$ and ${P}_{\phi}$, associated with
distinct quantum states $\psi$ and $\phi$, we have
$p(\lambda|P_{\psi})p(\lambda |P_{\phi})=0$ for all
$\lambda$.\footnote{Note that a better definition of the distinction
requires that $\int_\Lambda d\lambda
\sqrt{p(\lambda|P_{\psi})}\sqrt{p(\lambda |P_{\phi})}=0$. This
definition demands the vanishing of the \textit{classical fidelity},
rather than the product, of the probability distributions associated
with any pair of distinct pure quantum states. This refinement is
important for dealing with ontological models wherein the only pairs
of distributions that overlap do so on a set of measure zero.
Intuitively, one would not want to classify these as
$\psi$-epistemic, but only the fidelity-based definition does
justice to this intuition. This definition will not, however, be
needed here.\label{FNOTE:fidelity}} \label{DEF:psi_ontic}
\end{definition}

Hence, the epistemic states associated with distinct quantum states
are completely non-overlapping in a $\psi$-ontic model. In other
words, different quantum states pick out disjoint regions of
$\Lambda.$ The idea of a $\psi$-incomplete model that is also
$\psi$-ontic is illustrated schematically in part (b) of
Fig.~\ref{FIG:classes2}. Here, the ontic state space is
parameterized by $\psi$ (represented by a single axis) and a
supplementary hidden variable $\omega.$ \ The epistemic state
$p(\lambda|P_{\phi})$ representing a preparation procedure
associated with $\phi$ has the form of a Dirac delta function along
the $\psi$ axis, which guarantees the disjointness property for
epistemic states associated with distinct quantum states. Even if an
ontological model is presented to us in a form where it is not
obvious whether $\psi$ parameterizes $\Lambda$, by verifying that
the above definition is satisfied, one verifies that such a
parametrization can be found.

Another useful way of thinking about $\psi$-ontic models is that the
ontic state $\lambda$ `encodes' the quantum state $\psi$ because a
given $\lambda$ is only consistent with one choice of $\psi$.
Alternatively, we can see this encoding property as follows. By
Bayes' theorem, one infers that any $\psi$-ontic model satisfies
$p(P_{\psi}|\lambda)p(P_{\phi}|\lambda)=0$ for $\psi\neq\phi,$ which
implies that for every $\lambda,$ there exists some $\psi$ such that
$p(P_{\psi}|\lambda)=1$ and $p(P_{\phi}|\lambda)=0$ for all $\phi
\neq\psi.$

\begin{definition}
If an ontological model fails to be $\psi$-ontic, then it is said to
be \textbf{$\psi$-epistemic}.\label{DEF:psi_epistemic}
\end{definition}

It is worth spelling out what the failure of the $\psi$-ontic
property entails: there exists a pair of preparation procedures,
$P_{\psi}$ and ${P}_{\phi}$ and a $\lambda\in\Lambda$ such that
$p\left( \lambda|P_{\psi }\right) p\left( \lambda|P_{\phi}\right)
\neq{0}$, which is to say that the two epistemic states \emph{do
}overlap. Using Bayes' theorem we can equivalently formulate this
requirement as $\exists\:{P}_{\psi},{P}_{\phi },\lambda:p\left(
P_{\psi}|\lambda\right)  p\left(  P_{\phi}|\lambda\right) \neq{0},$
which asserts that the ontic state $\lambda$ is consistent with both
the quantum state $\psi$ and the quantum state $\phi.$ \ In a
$\psi$-epistemic model, multiple distinct quantum states are
consistent with the same state of reality -- the ontic state
$\lambda$ does not encode $\psi.$ It is in this sense that the
quantum state is judged epistemic in such models. This is
illustrated schematically in part (c) of Fig.~\ref{FIG:classes2},
where the ontic state space $\Lambda$ is one-dimensional, and
preparations associated with distinct $\psi$ are associated with
overlapping distributions on $\Lambda.$

Some comments are in order. The reader might well be wondering why
we do not admit that \emph{any} $\psi$-incomplete model is
`epistemic', simply because it associates a probability distribution
of nontrivial width over $\Lambda$ with each quantum state. We admit
that although it might be apt to say that $\psi$-incomplete
\textit{models} have an epistemic character, the question of
interest here is whether \emph{pure quantum states} have an
epistemic character. It is for this reason that we speak of whether
a model is `$\psi $-epistemic' rather than simply `epistemic'. By
our definitions, $\psi$ has an ontic character if and only if a
variation of $\psi$ implies a variation of reality and an epistemic
character if and only if a variation of $\psi$ does \emph{not}
necessarily imply a variation of reality.

\begin{figure}[t]
\includegraphics[scale=0.4]{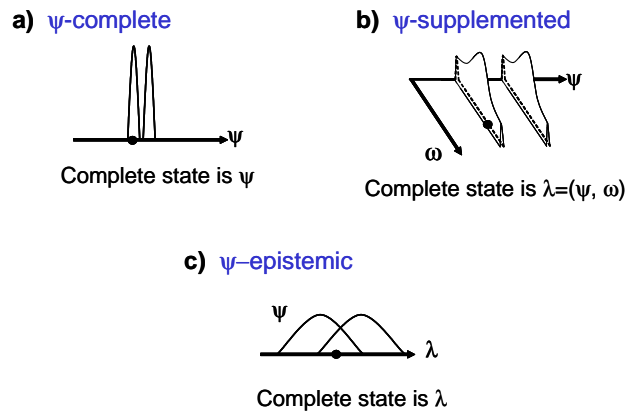}\caption{Schematic representation of how
probability distributions associated with $\psi$ are related in
\textbf{(a)} $\psi$-complete models, \textbf{(b)}
$\psi$-supplemented models and \textbf{(c)} $\psi$-epistemic models.
Note that the narrow gaussian shaped distributions in part
\textbf{(b)} denote an arbitrary distribution over the supplementary
variables $\omega$ combined with a dirac-delta function over the set
of
quantum states, $\psi$.}%
\label{FIG:classes2}%
\end{figure}

For any model we can specify a $\psi$-complete versus
$\psi$-incomplete and $\psi$-ontic versus $\psi$-epistemic
classification. At first sight, this suggests that there will be
four different types of ontological model. This impression is
mistaken however; there are only three different types of model
because one of the four combinations describes an empty set.
Specifically, if a model is $\psi$-complete, then it is also
$\psi$-ontic. This follows from the fact that if a model is $\psi
$-complete, then $p\left(  \lambda|P_{\psi}\right) =\delta\left(
\lambda-\lambda_{\psi}\right),$ where $\lambda_{\psi}$ is the ontic
state isomorphic to $\psi,$ and from the fact that $\delta\left(
\lambda -\lambda_{\psi}\right) \delta\left(
\lambda-\lambda_{\phi}\right)  =0$ for $\psi\neq\phi$.\

The contrapositive of this implication asserts that for the quantum
state to have an epistemic character, it cannot be a complete
description of reality. We have therefore proven:

\begin{lemma}
The following implications between properties of ontological models
hold\footnote{Implications such as $C_{1}\rightarrow{C}_{2}$ between
two classes $C_{1}$ and $C_{2}$ of ontological models should be read
as `any model in class $C_{1}$ is necessarily also in class
$C_{2}$'.}:
\[
\psi\text{-complete}\rightarrow\psi\text{-ontic,}
\]
and its negation,
\begin{equation}
\psi\text{-epistemic}\rightarrow\psi\text{-incomplete.}
\end{equation}
\label{LEM:psi_ontic}
\end{lemma}

So it is impossible for a model to be both $\psi$-complete and $\psi
$-epistemic. Given Lemma \ref{LEM:psi_ontic}, we can unambiguously
refer to models that are $\psi$-complete and $\psi$-ontic as simply
$\psi$-complete, and models that are $\psi $-incomplete and
$\psi$-epistemic as simply $\psi$-epistemic. The $\psi
$-supplemented models constitute the third category.

\begin{definition}
Ontological models that are $\psi$-incomplete and $\psi$-ontic will be
referred to as \textbf{$\psi$-supplemented}.\label{DEF:psi_ontic_plus}
\end{definition}

The classification of ontological models is summarized in
Fig.~\ref{FIG:classes2}.

\begin{figure}[t]
\includegraphics[scale=0.5]{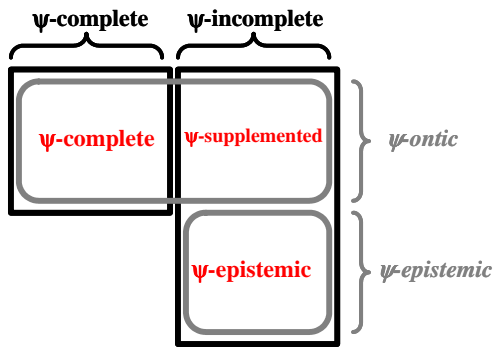}\caption{Two distinctions and the three
classes of ontological model that they define.}
\label{FIG:classes_chart}
\end{figure}

\subsection{Examples \label{SEC:example_models}}

We now provide examples from the literature of models that fall into
each class.

\subsubsection{The Beltrametti-Bugajski model \label{SEC:example_bb}}

The model of Beltrametti and Bugajski \cite{beltrametti_bugajski} is
essentially a thorough rendering of what most would refer to as an
orthodox interpretation of quantum mechanics.\footnote{Note,
however, that there are several versions of orthodoxy that differ in
their manner of treating measurements. The Beltrametti-Bugajski
model is distinguished by the fact that it fits within the framework
for ontological models we have outlined.} The ontic state space
postulated by the model is precisely the projective Hilbert space,
$\Lambda=\mathcal{PH}$, so that a system prepared in a quantum state
$\psi$ is associated with a sharp probability
distribution\footnote{Preparations which correspond to mixed quantum
states can be constructed as a convex sum of such sharp
distributions} over $\Lambda$,
\begin{equation}
p\left(  \lambda|\psi\right)  =\delta\left(  \lambda-\psi\right),
\end{equation}
where we are using $\psi$ interchangeably to label the Hilbert space
vector and to denote the ray spanned by this vector

The model posits that the different possible states of reality are
simply the different possible quantum states. \ It is therefore
$\psi$-complete by Definition \ref{DEF:psi_complete}. \ It remains
only to demonstrate how it reproduces the quantum statistics.

This is achieved by assuming that the probability of obtaining an
outcome $k$ of a measurement procedure $M$ depends
indeterministically on the system's ontic state $\lambda$ as
\begin{equation}
p\left(  k|M,\lambda\right)  =\text{tr}\left(  |\lambda\rangle\langle
\lambda|E_{k}\right),\label{BB1}%
\end{equation}
where $\left\vert \lambda\right\rangle \in\mathcal{H}$ denotes the
quantum state associated with $\lambda\in\mathcal{PH}$, and where
$\left\{ E_{k}\right\}$ is the POVM that quantum mechanics
associates with $M$. It follows that,
\begin{align}
\mathrm{Pr}\left(  k|M,\psi\right)   &  =\int_{\Lambda}{d}\lambda\text{ }%
{p}\left(  k|M,\lambda\right)  \text{ }p(\lambda|\psi)\nonumber\\
&  =\int_{\Lambda}{d}\lambda\text{ tr}\left(  |\lambda\rangle\langle
\lambda|E_{k}\right)  \text{ }\delta\left(  \lambda-\lambda_{\psi}\right)
\label{BB2}\\
&  =\text{tr}\left(  |\psi\rangle\langle\psi|E_{k}\right),\label{BB3}%
\end{align}
and so the quantum statistics are trivially reproduced.

If we restrict consideration to a system with a two dimensional
Hilbert space then $\Lambda$ is isomorphic to the Bloch sphere, so
that the ontic states are parameterized by the Bloch vectors of unit
length, which we denote by $\vec{\lambda}.$ \ The Bloch vector
associated with the Hilbert space ray $\psi$ is denoted $\vec{\psi}$
and is defined by $\left\vert \psi\right\rangle \left\langle
\psi\right\vert
=\frac{1}{2}I+\frac{1}{2}\vec{\psi}\cdot\vec{\sigma}$ where
$\vec{\sigma }=(\sigma_{x},\sigma_{y},\sigma_{z})$ denotes the
vector of Pauli matrices and $I$ denotes the identity operator.

If we furthermore consider $M$ to be a \textit{projective}
measurement, then it is associated with a projector-valued measure
$\{\left\vert \phi \right\rangle \left\langle \phi\right\vert
,\left\vert \phi^{\perp }\right\rangle \left\langle
\phi^{\perp}\right\vert \}$ or equivalently, an orthonormal basis
$\{\left\vert \phi\right\rangle ,\left\vert \phi^{\perp
}\right\rangle \}.$ \ It is convenient to denote the probability of
getting the $\phi$ outcome given ontic state $\vec{\lambda}$ simply
by $p(\phi |\vec{\lambda})$. \ Eq.~(\ref{BB1}) simplifies to,
\begin{align}
p(  \phi|\vec{\lambda})    & =|\left\langle \phi|\lambda
\right\rangle |^{2}\\
& =\frac{1}{2}\left(  1+\vec{\phi}\cdot\vec{\lambda}\right)
\label{BBif}.
\end{align}
The epistemic states and indicator functions for this case of the
Beltrametti-Bugajski model are illustrated schematically in
Fig.~\ref{FIG:bbmodel}.

\begin{figure}[t]
\includegraphics[scale=0.6]{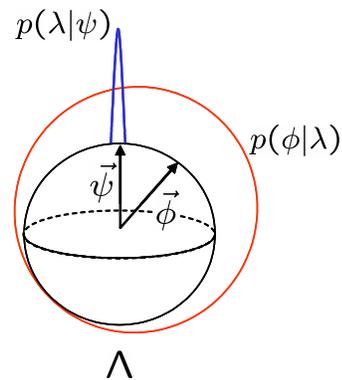}\caption{Illustration of the epistemic
states and indicator functions in the Beltrametti-Bugajski model.}
\label{FIG:bbmodel}
\end{figure}

\subsubsection{The Bell-Mermin model \label{SEC:example_bellmermin}}

We now present an ontological model for a two dimensional Hilbert
space that is originally due to Bell \cite{Bell_probhv} and was
later adapted into a more intuitive form by Mermin
\cite{Mermin_bell}. \

The model employs an ontic state space $\Lambda$ that is a Cartesian
product of a pair of state spaces,
$\Lambda=\Lambda^{\prime}\times\Lambda ^{\prime\prime}$. Each of
$\Lambda^{\prime}$ and $\Lambda^{\prime\prime}$ is isomorphic to the
unit sphere. It follows that there are two variables required to
specify the systems total ontic state,
$\vec{\lambda}^{\prime}\in\Lambda^{\prime}$ and
$\vec{\lambda}^{\prime\prime}\in\Lambda^{\prime\prime}$. A system
prepared according to quantum state $\psi$ is assumed to be
described by a product
distribution on $\Lambda^{\prime}\times\Lambda^{\prime\prime}$,%

\begin{equation}
p(\vec{\lambda}^{\prime},\vec{\lambda}^{\prime\prime}|\psi)=p(\vec{\lambda
}^{\prime}|\psi)p(\vec{\lambda}^{\prime\prime}|\psi).
\end{equation}
The distribution over $\vec{\lambda}^{\prime}$ is a Dirac delta
function centered on $\vec{\psi},$ that is,
$p(\vec{\lambda}^{\prime}|\psi)=\delta(\vec{\lambda}^{\prime}-\vec{\psi})$.
The distribution over
$\vec{\lambda}^{\prime\prime}\in\Lambda^{\prime\prime}$ is uniform
over the unit sphere,
$p(\vec{\lambda}^{\prime\prime}|\psi)=\frac{1}{4\pi}$, independent
of $\psi $. These epistemic states are illustrated in
Fig.~\ref{FIG:bmbs1and2}. Consequently,

\begin{equation}
p(\vec{\lambda}^{\prime},\vec{\lambda}^{\prime\prime}|\psi)=\frac{1}{4\pi
}\delta(\vec{\lambda}^{\prime}-\vec{\psi}).
\end{equation}

Suppose now that we wish to perform a projective measurement
associated with the basis $\{\left\vert \phi\right\rangle
,\left\vert \phi^{\perp }\right\rangle \}$. The Bell-Mermin model
posits that the $\phi$ outcome will occur if and only if the vector
$\vec{\lambda}^{\prime}+\vec{\lambda }^{\prime\prime}$ has a
positive inner product with the Bloch vector $\vec{\phi}.$ \ This
measurement is therefore associated with the indicator function,
\begin{equation}
p(\phi|\vec{\lambda}^{\prime},\vec{\lambda}^{\prime\prime})=\Theta(\vec{\phi
}\cdot(\vec{\lambda}^{\prime}+\vec{\lambda}^{\prime\prime})),
\label{BMif}
\end{equation}

where $\Theta$ is the Heaviside step function defined by
\begin{align*}
\Theta(x)  & =1\text{ if }x>0\\
& =0\text{ if }x\le 0.
\end{align*}

The Bell-Mermin model's predictions for $p(\phi|\psi)$ (calculated
as the overlap of the epistemic distributions from
Fig.~\ref{FIG:bmbs1and2} with the indicator function defined in
(\ref{BMif})) successfully reproduce the quantum mechanical Born
rule,
\begin{align}
p(\phi|\psi)  &
=\frac{1}{4\pi}\int\!\!\!\int{d}\Lambda^{\prime}{d}\Lambda
^{\prime\prime}\:\delta(\vec{\lambda}^{\prime}-\vec{\lambda}^{\prime}_{\psi
})\:\Theta(\vec{\lambda}^{\prime}_{\phi}\cdot(\vec{\lambda}^{\prime}%
+\vec{\lambda}^{\prime\prime}))\nonumber\\
&  =\frac{1}{2}\left(  1+\vec{\lambda}^{\prime}_{\phi}\cdot\vec{\lambda
}^{\prime}_{\psi}\right) \nonumber\\
&  =\left|  \langle\psi|\phi\rangle\right|  ^{2}.
\end{align}
We can see immediately that the Bell-Mermin model is
$\psi$-incomplete because
$\Lambda=\Lambda^{\prime}\times\Lambda^{\prime\prime}\neq\mathcal{PH}$.
Furthermore,
\begin{align}
p(\lambda|\psi)p(\lambda|\phi) &  =p(\vec{\lambda}^{\prime},\vec{\lambda
}^{\prime\prime}|\psi)p(\vec{\lambda}^{\prime},\vec{\lambda}^{\prime\prime
}|\phi)\nonumber\\
&  =\frac{1}{16\pi^{2}}\delta(\vec{\lambda}^{\prime}-\vec{\psi})\delta
(\vec{\lambda}^{\prime}-\vec{\phi})\nonumber\\
&  =0\:\:\:\text{if}\:\:\:\psi\neq\phi,
\end{align}
implying that the Bell-Mermin model is $\psi$-ontic. Recalling Definition
\ref{DEF:psi_ontic_plus} we conclude that this model falls into the class
$\psi$-supplemented.

\begin{figure}[t]
\includegraphics[scale=0.55]{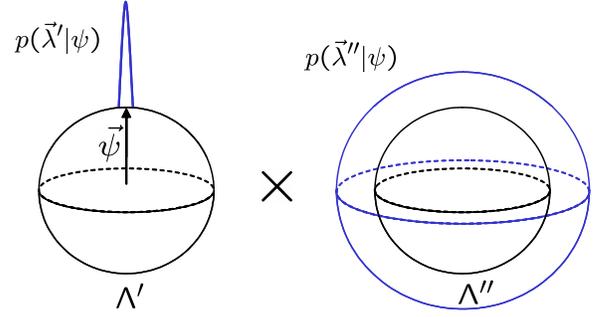}\caption{Illustration of the epistemic
states in the Bell-Mermin model.} \label{FIG:bmbs1and2}
\end{figure}

Because the ontic state space of this model is four dimensional, it
is difficult to illustrate it in a figure. We can present the
distributions over $\vec{\lambda}^{\prime}$ and
$\vec{\lambda}^{\prime\prime}$ on separate unit spheres, as in
Fig.~\ref{FIG:bmbs1and2} , but the indicator functions cannot be
presented in this way.

\subsubsection{The Kochen-Specker model \label{SEC:example_ksmodel}}

As our final example, we consider a model for a two-dimensional
Hilbert space due to Kochen and Specker \cite{Ks}. The ontic state
space $\Lambda$ is taken to be the unit sphere, and a quantum state
$\psi$ is associated with the probability distribution,

\begin{equation}
p(\lambda|\psi)=\frac{1}{\pi}\Theta(\vec{\psi}\cdot\vec{\lambda})\:\vec{\psi
}\cdot\vec{\lambda}, \label{mu_kochen_specker}
\end{equation}
where $\vec{\psi}$ is the Bloch vector corresponding to the quantum
state $\psi$. It assigns the value $\cos{\theta}$ to all points an
angle $\theta<\frac{\pi}{2}$ from $\psi $, and the value zero to
points with $\theta>\frac{\pi}{2}$. This is illustrated in
Fig.~\ref{FIG:ksmodel}.

Upon implementing a measurement procedure $M$ associated with a
projector $|\phi\rangle\langle\phi|$ a positive outcome will occur
if the ontic state $\vec{\lambda}$ of the system lies in the
hemisphere centered on $\vec{\phi}$, i.e.,
\begin{equation}
p(\phi|\lambda)=\Theta(\vec{\phi}\cdot\vec{\lambda}).
\end{equation}
It can be checked that the overlaps of $p(\lambda|\psi)$ and
$p(\phi|\lambda)$ then reproduce the required quantum statistics,
\begin{align}
p(\phi|\psi)  & =\int{d}\lambda\:\frac{1}{\pi}\Theta(\vec{\psi}\cdot
\vec{\lambda})\Theta(\vec{\phi}\cdot\vec{\lambda})\:\vec{\psi}\cdot
\vec{\lambda}\nonumber\\
&  =\frac{1}{2}(1+\vec{\psi}\cdot\vec{\phi})\nonumber\\
&  =\left|  \langle\psi|\phi\rangle\right|  ^{2}.
\end{align}

Referring to Definition \ref{DEF:psi_incomplete} we see that this
model is $\psi$-incomplete, since although $\Lambda$ is isomorphic
to the system's projective Hilbert space,
Eq.~(\ref{mu_kochen_specker}) implies that the model associates
\textit{non-sharp} distributions with quantum states. Furthermore,
\[
p(\lambda|\psi)p(\lambda|\phi)=\frac{1}{\pi^{2}}\Theta(\vec{\psi}\cdot
\vec{\lambda})\;\Theta(\vec{\phi}\cdot\vec{\lambda})\;\vec{\psi}\cdot
\vec{\lambda}\;\vec{\phi}\cdot\vec{\lambda},
\]
is nonzero for nonorthogonal $\phi$ and $\psi,$ showing, via
Definition \ref{DEF:psi_epistemic}, that the Kochen-Specker model is
$\psi$-epistemic.

\begin{figure}[t]
\includegraphics[scale=0.6]{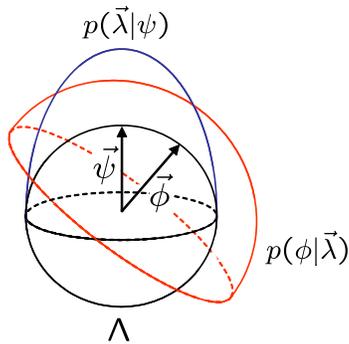}\caption{Illustration of the epistemic
states and indicator functions of the Kochen-Specker model.}
\label{FIG:ksmodel}
\end{figure}

\subsubsection{Connections between the models}

It is not too difficult to see that the Bell-Mermin model is simply
the Beltrametti-Bugajski model supplemented by a hidden variable
that uniquely determines the outcomes of all projective
measurements. \ We need only note that within the Bell-Mermin model,
the probability of obtaining the measurement outcome $\phi$ given
$\vec{\lambda}^{\prime}$ (i.e. not conditioning on the supplementary
hidden variable $\vec{\lambda}^{\prime\prime }$) is,
\begin{align*}
p(\phi|\vec{\lambda}^{\prime})  &
=\int_{\Lambda}d\vec{\lambda}^{\prime\prime
}p(\phi|\vec{\lambda}^{\prime},\vec{\lambda}^{\prime\prime})p(\vec{\lambda
}^{\prime\prime})\\
& =\int_{\Lambda}d\vec{\lambda}^{\prime\prime}\Theta(\vec{\phi}\cdot
(\vec{\lambda}^{\prime}+\vec{\lambda}^{\prime\prime}))\frac{1}{4\pi}\\
& =\frac{1}{2}\left(  1+\vec{\phi}\cdot\vec{\lambda}^{\prime}\right)
,
\end{align*}
which is precisely the indicator function of the
Beltrametti-Bugajski model, Eq.~(\ref{BBif}). So, whereas in the
Beltrametti-Bugajski model, outcomes that are not determined
uniquely by $\vec{\lambda}^{\prime}$ (i.e. for which $0<p(\phi
|\vec{\lambda}^{\prime})<1$) are deemed to be objectively
indeterministic, in the Bell-Mermin model this indeterminism is
presumed to be merely epistemic, resulting from ignorance of the
value of the supplementary hidden variable $\vec{\lambda
}^{\prime\prime}.$ Note that although the Bell-Mermin model
eliminates the objective indeterminism of the Beltrametti-Bugajski
model, it pays a price in ontological economy -- the dimensionality
of the ontic state space is doubled.

Furthermore, there is a strong connection, previously unnoticed,
between the Bell-Mermin model and the Kochen-Specker model. Although
the ontic state is specified by two variables, $\vec{\lambda}'$ and
$\vec{\lambda}''$, in the Bell-Mermin model, the indicator functions
for projective measurements, presented in Eq.~(\ref{BMif}), depend
only on $\vec{\lambda}^{\prime}+\vec{\lambda }^{\prime\prime}$. It
follows that if one re-parameterizes the ontic state space by the
pair of vectors
$\vec{u}=\vec{\lambda}^{\prime}+\vec{\lambda}^{\prime\prime}$ and
$\vec{v}=\vec{\lambda}^{\prime}-\vec{\lambda}^{\prime\prime}$, then
the indicator functions depend only on $\vec{u}$. Consequently, the
only aspect of the epistemic state that is significant for
calculating operational predictions is the marginal
$p(\vec{u}|\psi)$. This is calculated to be,
\begin{align}
p(\vec{u}|\psi)  & =\int d\vec{v}\text{ }p(\vec{u},\vec{v}|\psi)\\
\nonumber & =\int
d\vec{v}\frac{1}{4\pi}\delta(\frac{1}{2}\vec{u}+\frac{1}{2}\vec
{v}-\vec{\psi})\\ \nonumber
&=\frac{2}{\pi}\Theta(\vec{\psi}\cdot\vec{u})\text{
}\vec{\psi}\cdot\vec{u}.
\end{align}
But, on normalizing the vector $\vec{u}$ to lie on the unit sphere,
this is precisely the form of the epistemic state posited by the
Kochen-Specker model, Eq.~(\ref{mu_kochen_specker}), with $\vec{u}$
substituted for $\vec{\lambda}$.

It follows that the Kochen-Specker model is simply the Bell-Mermin
model with the variable $\vec{v}$ eliminated, so that the variable
$\vec{u}$ completely specifies the ontic state. (Reducing the ontic
state space in this way leaves the empirical predictions of the
model intact because these did not depend on $\vec{v}$.)

A methodological principle that is often adopted in the construction
of physical theories is that one should not posit unnecessary
ontological structure. Appealing to Occam's razor in the present
context would lead one naturally to judge the variable $\vec{v}$ to
be un-physical, akin to a gauge degree of freedom, and to thereby
favor the minimalist ontological structure posited by the
Kochen-Specker model over that of the Bell-Mermin model.

We see, therefore, that the price in ontological overhead that was
paid by the Bell-Mermin model to eliminate objective indeterminism
from the Beltrametti-Bugajski model did not need to be paid. The
Kochen-Specker model renders the indeterminism epistemic without any
increase in the size of the ontic state space.

It is interesting to note that starting from the orthodox model of
Beltrametti and Bugajski for two dimensional Hilbert spaces, if one
successively enforces (1) a principle that any indeterminism must be
epistemic rather than objective, and (2) a principle that any
gauge-like degrees of freedom must be eliminated as un-physical, one
arrives at the $\psi$-epistemic model of Kochen and Specker. One is
led to wonder whether such a procedure might be applied to
ontological models of quantum theory in higher dimensional Hilbert
spaces.

This concludes our discussion of the classification scheme for
ontological models. We now turn our attention to the question of how
these classes fare on the issue of locality.

\begin{figure}[t]
\includegraphics[scale=0.6]{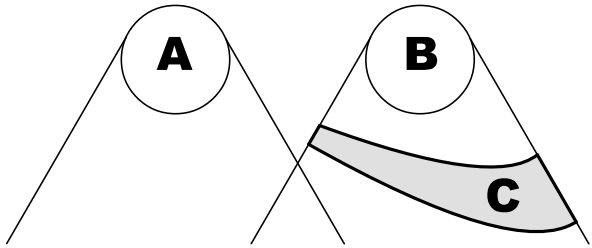}\caption{Space-time regions used in the definition of local causality proposed by Bell \cite{Bell_cuisine}.}
\label{FIG:bell_fig}
\end{figure}

\section{Locality in ontological models \label{SEC:locality}}

A necessary component of any sensible notion of locality is
separability, which we define as follows.
\begin{definition}
Suppose a region $R$ can be divided into local regions $R_{1},R_{2},...,R_{n}%
$. An ontological model is said to be \emph{separable} (denoted $S$) only if
the ontic state space $\Lambda_{R}$ of region $R$ is the Cartesian product of
the ontic state spaces $\Lambda_{R_{i}}$ of the regions $R_{i},$%
\[
\Lambda_{R}=\Lambda_{R_{1}}\times\Lambda_{R_{1}}\times\dots\times
\Lambda_{R_{n}}.
\]
\label{DEF:separable}
\end{definition}
The assumption of separability is made, for instance, by Bell when
he restricts his attention to theories of \textit{local} beables.
These are variables parameterizing the ontic state space ``which
(unlike for example the total energy) can be assigned to some
bounded space-time region'' \cite{bell_quoteonseparability}.

Separability is generally not considered to be a sufficient
condition for locality. An additional notion of locality, famously
made precise by Bell \cite{Bell_cuisine,bell_beables}, appeals to
the causal structure of relativistic theories.  The definition
appeals to the space-time regions defined in Fig.
\ref{FIG:bell_fig}.  Regions $A$ and $B$ are presumed to be
space-like separated.
\begin{definition}
A separable ontological model is \emph{locally causal} (LC) if and
only if the probabilities of events in space-time region $B$ are
unaltered by specification of events in space-time region $A$, when
one is already given a complete specification of the events in a
space-time region $C$ that screens off $B$ from the intersection of
the backward light cones of $A$ and $B$.
\label{DEF:locally_causal}
\end{definition}
Local causality can be expressed as
\begin{equation}
p(B|A,\lambda_{C})=p(B|\lambda_{C}),
\end{equation}
where $B$ is a proposition about events occurring in region $B,$
$\lambda _{C}$ is the ontic state of space-time region C (recalling
that the ontic state of a system is a complete specification of the
properties of that system), and $A$ is a proposition about events in
region $A$.

Finally, we define locality to be the conjunction of these two
notions.
\begin{definition}
An ontological model is \emph{local} (L) if and only if it is
separable and locally causal. \label{DEF:bell_local}
\end{definition}

Given that the Hilbert space associated with a pair of distinct
regions of space (or a pair of systems confined to distinct regions)
is the tensor product of the Hilbert spaces associated with each
region (or system), rather than the Cartesian product, we deduce
directly from Definitions \ref{DEF:psi_complete} and
\ref{DEF:separable} that $\psi$-complete models are not separable,
and consequently not local,\footnote{Some might argue that the ontic
state space of a system should include the mixed quantum states.
However, even if the ontic state space of a system were taken to be
the convex hull of the projective Hilbert space for that system, the
condition of separability would still not be satisfied because the
Cartesian product of the ontic state spaces of two systems would not
contain any correlated quantum states.}
\begin{equation}
\psi\text{-complete}\implies\lnot\text{S}\implies\lnot\text{L}.
\end{equation}
So, one needn't even test whether $\psi$-complete models are locally
causal, given that they fail to even exhibit separability, which is
a prerequisite to making sense of the notion of local causality.

There is in fact good evidence that this kind of reasoning captures
Einstein's earliest misgivings about quantum theory. Already in
1926, Einstein judges Schr\"{o}dinger's wave mechanics to be
``altogether too primitive'' \cite{Einstein_ehrenfest}. Howard has
argued convincingly that the significant issue for Einstein, even in
those early days, was separability \cite{Howard_einst_long}. For
instance, in order to describe multi-particle systems,
Schr\"{o}dinger had replaced de Broglie's waves in 3-space with
waves in configuration space, and had abandoned the notion of
particle trajectories (thereby endorsing a $\psi$-complete view).
But Einstein was dubious of this move: ``The field in a
many-dimensional coordinate space does not smell like something
real''\cite{Einstein_ehrenfest2}, and ``If only the undulatory
fields introduced there could be transplanted from the n-dimensional
coordinate space to the 3 or 4
dimensional!''\cite{Einstein_sommerfield}.

Nonetheless, even if one ignores the non-separability of entangled
quantum states, it is straightforward to show that the manner in
which such states are updated after local measurements implies a
failure of local causality if one adopts a $\psi$-complete model.
Einstein first made this argument later in 1927, as we shall see in
Sec.~\ref{SEC:historical_1927}.

\subsection{$\psi$-ontic models of quantum theory are \textbf{nonlocal}
\label{SEC:locality_theorem}}

We now demonstrate that there exists a very simple argument
establishing that \emph{all} $\psi$-ontic models (not just those
that are $\psi$-complete) must violate locality. The argument
constitutes a ``nonlocality theorem'' that is stronger than
Einstein's 1927 argument but weaker than Bell's theorem. In the next
section, we shall argue that it is in fact the content of Einstein's
1935 argument for incompleteness (the argument appearing in his
correspondence with Schr\"{o}dinger, not the EPR paper) and we shall
explore what light is thereby shed on his interpretational stance.
For now, however, we shall simply present the argument in the
clearest possible fashion.

Consider two separated parties, Alice and Bob, who each hold one
member of a pair of two-level quantum systems prepared in the
maximally entangled state $\left\vert \psi^{+}\right\rangle =\left(
\left\vert 0\right\rangle \left\vert 1\right\rangle +\left\vert
1\right\rangle \left\vert 0\right\rangle \right)/\sqrt{2}$. If Alice
chooses to implement a measurement $M_{01}$ associated with the
basis $\{\left\vert 0\right\rangle ,\left\vert 1\right\rangle \}$,
then depending on whether she obtains outcome $0$ or $1$, she
updates the quantum state of Bob's system to $\left\vert
0\right\rangle $ or $\left\vert 1\right\rangle $ respectively (these
occur with equal probability). On the other hand, if she implements
a measurement $M_{\pm}$ associated with the basis $\left\{
|+\rangle,|-\rangle\right\},$ where $|\pm\rangle = (|0\rangle \pm
|1\rangle)/\sqrt{2}$, then she updates the quantum state of Bob's
system to $\left\vert +\right\rangle $ or $\left\vert -\right\rangle
$ depending on her outcome. Although Alice cannot control which
individual pure quantum state will describe Bob's system, she can
choose which of two disjoint sets, $\{\left\vert 0\right\rangle
,\left\vert 1\right\rangle \}$ or $\{\left\vert +\right\rangle
,\left\vert -\right\rangle \},$ it will belong to. Schr\"{o}dinger
described this effect as `steering' Bob's state
\cite{schroed_steer}.

This steering phenomenon allows us to prove the following
theorem\footnote{Note that one might suppose that the conclusion of
Theorem \ref{THRM:psi_ont_nonlocal} can be arrived at more simply by
the line of reasoning,
$\psi$-ontic$\rightarrow\lnot{S}\rightarrow\lnot{L}$. However, it is
not clear whether the ability to supplement a $\psi$-ontic model
with `hidden variables' allows one to alleviate a violation of
separability within $\psi$-ontic models.}.
\begin{theorem}
Any $\psi$-ontic ontological model that reproduces the quantum
statistics (QSTAT) violates locality,\footnote{Note that no notion
of `realism' appears in our implication.  This is because there is
no sense in which there is an assumption of realism that could be
abandoned while salvaging locality.  There \emph{is} a notion of
realism at play when we grant that experimental procedures prepare
and measure properties of systems, but it is a \emph{prerequisite}
to making sense of the notion of locality. Norsen has emphasized
this point \cite{norsen_blrealism,norsen_againstrealism}.}
\[
\psi\text{-ontic}\wedge\text{QSTAT}\rightarrow\lnot\text{L.}
\]
\label{THRM:psi_ont_nonlocal}
\end{theorem}
\begin{proof}
The measurements that Alice performs can be understood as `remote
preparations' of Bob's system (recall from Sec.~\ref{SEC:om_intro}
that a preparation is simply a list of experimental instructions and
therefore need not involve a direct interaction with the system
being prepared). Denote by $P_{0}$ and $P_{1}$ the remote
preparations corresponding to Alice measuring $M_{01}$ and obtaining
the $0$ and $1$ outcomes respectively (these preparations are
associated with the states $\left\vert 0\right\rangle $ and
$\left\vert 1\right\rangle $ of Bob's system). Let $P_{+}$ and
$P_{-}$ be defined similarly. Finally, denote by $P_{01}$ the remote
preparation that results from a measurement of $M_{01}$ but wherein
one does not condition on the outcome, and similarly for $P_{\pm}.$
\ Given these definitions, we can infer that,
\begin{align}
p(\lambda|P_{01})  &
=\frac{1}{2}p(\lambda|P_{0})+\frac{1}{2}p(\lambda
|P_{1})\, ,\label{decomps_in_om_01}\\
p(\lambda|P_{\pm})  &
=\frac{1}{2}p(\lambda|P_{+})+\frac{1}{2}p(\lambda
|P_{-})\, ,\label{decomps_in_om_pm}%
\end{align}
where $\lambda$ is the ontic state of Bob's system, which is
well-defined by virtue of the assumption of separability.
Eqs.~(\ref{decomps_in_om_01}) and (\ref{decomps_in_om_pm}) are
justified by noting that the probability one assigns to $\lambda$ in
the unconditioned case is simply the weighted sum of the probability
one assigns in each of the conditioned cases, where the weights are
the probabilities for each condition to hold \cite{Spekkens_con}.

The proof is by contradiction. The assumption of local causality
implies that the probabilities for Bob's system being in various
ontic states are independent of the measurement that Alice performs.
Consequently,\footnote{Note that an assumption of no superluminal
signalling is not sufficient to obtain Eq.~(\ref{bell_locality})
because $p(\lambda|P_{01})$ and $p(\lambda|P_{\pm})$ could change
non-locally, but in such a way that every indicator function on
system $B$ that corresponds to a possible measurement is unable to
distinguish $p(\lambda|P_{01})$ from $p(\lambda|P_{\pm})$ despite
their differences.}
\begin{equation}
p(\lambda|P_{01})=p(\lambda|P_{\pm}).\label{bell_locality}
\end{equation}
Multiplying together Eqs.~(\ref{decomps_in_om_01}) and (\ref{decomps_in_om_pm}%
) and making use of Eq.~(\ref{bell_locality}), we obtain,
\begin{align}
4\,p\left(\lambda|P_{01}\right)^{2} &= p\left(\lambda|P_{+}\right)
p\left(\lambda|P_{0}\right) + p\left(\lambda|P_{+}\right)
p\left(\lambda|P_{1}\right)
  \nonumber\\
& + p\left(\lambda|P_{-}\right)  p\left(\lambda|P_{0}\right)
+p\left(\lambda|P_{-}\right)  p\left(\lambda|P_{1}\right).
\end{align}
Therefore, for any $\lambda$ within the support of $p\left(  \lambda
|P_{01}\right)$ (a non-empty set), we must have,
\begin{align}
&  p\left(\lambda|P_{+}\right)  p\left(\lambda|P_{0}\right) +p\left(
\lambda|P_{+}\right)  p\left(\lambda|P_{1}\right)  \nonumber\\
&  +p\left(\lambda|P_{-}\right)  p\left(\lambda|P_{0}\right)
+p\left(\lambda|P_{-}\right)  p\left(\lambda|P_{1}\right)  >0,
\end{align}
which requires that at least one of the following inequalities be
satisfied,
\begin{align}
p\left(\lambda|P_{+}\right)  p\left(\lambda|P_{0}\right)   &
>0,\nonumber\\
p\left(\lambda|P_{+}\right)  p\left(\lambda|P_{1}\right)   &
>0,\nonumber\\
p\left(\lambda|P_{-}\right)  p\left(\lambda|P_{0}\right)   &
>0,\nonumber\\
p\left(\lambda|P_{-}\right)  p\left(\lambda|P_{1}\right)   &
>0.
\end{align}
It follows that there exists at least one pair of distinct quantum
states (either $\left\vert +\right\rangle ,\left\vert 0\right\rangle
$ or $\left\vert +\right\rangle ,\left\vert 1\right\rangle $ or
$\left\vert -\right\rangle ,\left\vert 0\right\rangle $ or
$\left\vert -\right\rangle ,\left\vert 1\right\rangle )$ such that
the epistemic states associated with them are overlapping on the
ontic state space. By Definition \ref{DEF:psi_epistemic}, we infer
that the ontological model must therefore be $\psi$-epistemic.
\end{proof}

\section{Reassessing Einstein's arguments for incompleteness \label{SEC:historical}}

\subsection{The EPR incompleteness argument}

It is well known that Einstein disputed the claim that the quantum
state represented a complete description of reality on the grounds
that such a view implied a failure of locality. Einstein's views on
the matter are often assumed to be well represented by the contents
of the EPR paper \cite{EPR}. There is, however, strong evidence
suggesting that this is far from the truth. Einstein describes his
part in the paper in a letter to Schr\"{o}dinger dated June 19, 1935
\cite{EtoS1935}:
\begin{quote}
``For reasons of language this [paper] was written by Podolsky after
many discussions. But still it has not come out as well as I really
wanted; on the contrary, the main point was, so to speak, buried by
the erudition.''
\end{quote}
Fine describes well the implications of these comments:
\cite{FineEcritique}.
\begin{quote}
``I think we should take in the message of these few words: Einstein did not
write the paper, Podolsky did, and somehow the central point was obscured. No
doubt Podolsky (of Russian origin) would have found it natural to leave the
definite article out of the title [Can quantum mechanical description be
considered complete?]. Moreover the logically opaque structure of the piece is
uncharacteristic of Einstein's thought and writing. There are no earlier
drafts of this article among Einstein's papers and no correspondence or other
evidence that I have been able to find which would settle the question as to
whether Einstein saw a draft of the paper before it was published. Podolsky
left Princeton for California at about the time of submission and it could
well be that, authorized by Einstein, he actually composed it on his own.''
\end{quote}

A more accurate picture of Einstein's views is achieved by looking
to his own publications and his correspondence. Although it is not
widely known, Einstein presented a simple argument for
incompleteness at the 1927 Solvay conference. Also, in the letter to
Schr\"{o}dinger that we quote above, Einstein gives his own argument
for incompleteness, which makes use of a similar
\textit{gedankenexperiment} to the one described in the EPR paper,
but has a significantly different logical structure.

Before turning to the details of these two arguments, we summarize
the time-line of their presentation relative to EPR,
\begin{itemize}
\item \textbf{October 1927:} Einstein presents an incompleteness argument at
the Solvay conference \cite{bac_valentini}.
\item \textbf{May 1935:} The EPR argument for incompleteness is published \cite{EPR}.
\item \textbf{June 1935:} Einstein presents an incompleteness argument,
differing substantially from the EPR argument, in his correspondence
with Schr\"{o}dinger \cite{EtoS1935}. (This first appears in print
in March 1936 \cite{EPhysik}.) We will refer to this as Einstein's
1935 argument, not to be confused with the conceptually distinct EPR
argument from the same year.
\end{itemize}

\subsection{Einstein's 1927 incompleteness argument\label{SEC:historical_1927}}

Einstein's first public argument for the incompleteness of quantum
mechanics was presented during the general discussion at the 1927
Solvay conference \cite{bac_valentini}. Einstein considered a
\textit{gedankenexperiment} in which electron wave-functions are
diffracted through a small opening, so that they then impinge upon a
hemispherical screen, as illustrated in Fig.~\ref{FIG:hemisphere}.
\begin{figure}[t]
\includegraphics[scale=0.8]{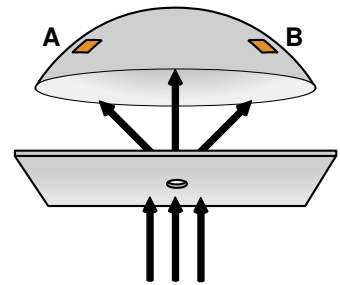}
\caption{Einstein's 1927 Gedankenexperiment, in which a single
particle wavefunction (blue) diffracts at a small opening (bottom)
before impinging upon a hemispherical detector (top). According to
quantum mechanics, the probability of a double detection at two
distinct regions $A$ and $B$ of the detector is zero.}
\label{FIG:hemisphere}
\end{figure}
He noted that \cite{E1927},
\begin{quote}
\textquotedblleft The scattered wave moving towards [the screen]
does not show any preferred direction. If $|\psi|^{2}$ were simply
regarded as the probability that at a certain point a given particle
is found at a given time, it could happen that the \textit{same}
elementary process produces an action in \textit{two or several}
places of the screen. But the interpretation, according to which
$|\psi|^{2}$ expresses the probability that \textit{this} particle
is found at a given point, assumes an entirely peculiar mechanism of
action at a distance which prevents the wave continuously
distributed in space from producing an action in \textit{two} places
on the screen.\textquotedblright
\end{quote}

Norsen has presented the essence of this argument in an elegant
form\footnote{The form of the argument is chosen to parallel the
form of Bell's argument in order to make evident the hypocrisy of a
widespread tendency among commentators to praise Bell's reasoning
while rejecting Einstein's.} that we reproduce here
\cite{norsenboxes}. Consider two points $A$ and $B$ on the screen
and denote by $1_{A}$ and $0_{A}$ respectively the cases where there
is or isn't an electron detected at $A$ (and similarly for $B$). We
take the initial quantum state of the electron to be of the form,
\begin{equation}
|\psi\rangle=\frac{1}{\sqrt{2}}(|A\rangle + |B\rangle),
\end{equation}
where $|A(B)\rangle$ is the quantum state that leads to an electron
detection at $A(B)$. Now suppose that one considers an ontological
model of the scenario, employing ontic states $\lambda\in\Lambda$.
Then the probability of obtaining a simultaneous detection at both
sites $A$ and $B$ is given by $p(1_{A}\wedge
{1}_{B}|\lambda)=p(1_{A}|\lambda)p(1_{B}|1_{A},\lambda)$. Suppose
furthermore that the model describing these events is assumed to be
local, then we can write $p(1_{B}|1_{A},\lambda)=p(1_{B}|\lambda)$
and thus $p(1_{A}\wedge
{1}_{B}|\lambda)=p(1_{A}|\lambda)p(1_{B}|\lambda)$. If the model is
taken to satisfy $\psi$-completeness then $\lambda=\psi$, and we
infer that,
\begin{equation}
p(1_{A}\wedge{1}_{B}|\psi)=p(1_{A}|\psi)p(1_{B}|\psi).
\label{probs_local}
\end{equation}
Inserting the quantum mechanical predictions $p(1_{A}|\psi)=p(1_{B}
|\psi)=\frac{1}{2}$, we obtain
$p(1_{A}\wedge{1}_{B}|\psi)=\frac{1}{4}$, which entails a nonzero
probability for simultaneous detections at both $A$ and $B$, in
stark contradiction with what is predicted by quantum mechanics.

Hence the logical structure of this rendition of Einstein's 1927
argument is that
$\text{L}\wedge\text{QSTAT}\wedge\psi\text{-complete}
\rightarrow\text{contradiction}$, i.e., that,
\begin{equation}
\text{L}\wedge\text{QSTAT}\rightarrow\psi\text{-incomplete}.
\label{E1927_logical}
\end{equation}

Note that, unlike the 1935 argument to which we shall turn in the
next section, the 1927 argument cannot be used to show locality to
be at odds with more general $\psi$-ontic models because if $\psi$
is supplemented with a hidden variable $\omega$, then the complete
description of the system is $\lambda=(\psi,\omega)$, and
Eq.~(\ref{probs_local}) is replaced by,
\begin{equation}
p(1_{A}\wedge{1}_{B}|\psi,\omega)=p(1_{A}|\psi,\omega)p(1_{B}|\psi,\omega).
\end{equation}
Because there is no reason to assume that
$p(1_{A}|\psi,\omega)=p(1_{A}|\psi)$ nor that
$p(1_{B}|\psi,\omega)=p(1_{B}|\psi)$ (conditioning on the hidden
variable will in general change the probability of detection), one
can no longer infer a nonzero probability for simultaneous
detections at both $A$ and $B$, and the contradiction is blocked.

\subsection{Einstein's 1935 incompleteness argument\label{SEC:historical_1935}}

In his 1935 correspondence with Schr\"{o}dinger, after noting that
the EPR paper did not do justice to his views, Einstein presents a
different version of the argument for incompleteness. The argument
differs markedly from that of the EPR paper from the very outset by
adopting a different notion of completeness \cite{EtoS1935},

\begin{quote}
``[...] one would like to say the following: $\psi$ is correlated
one-to-one with the real state of the real system. [...] If this
works, then I speak of a complete description of reality by the
theory. But if such an interpretation is not feasible, I call the
theoretical description `incomplete'.''
\end{quote}

It is quite clear that by `real state of the real system', Einstein
is referring to the ontic state pertaining to a system. Bearing this
in mind, his definition of completeness can be identified as
precisely our notion of $\psi $-completeness given in Definition
\ref{DEF:psi_complete}. Einstein then re-iterates to Schr\"{o}dinger
the beginning of the EPR argument, starting by considering a joint
system $(AB)$ to be prepared in an entangled state by some
`collision' between the subsystems $A$ and $B$. He then emphasizes
(what we would now call) the `steering phenomenon' by noting how a
choice of measurement on $A$ can result in the subsystem $B$ being
described by one of two quantum states $\psi_B$ or
$\psi_{\underline{B}}$.

Einstein then uses this scenario to derive his preferred proof of
incompleteness,

\begin{quote}
``Now what is essential is exclusively that $\psi_{B}$ and
$\psi_{\underline{B}}$ are in general different from one another. I
assert that this difference is incompatible with the hypothesis that
the description is correlated one-to-one with the physical reality
(the real state). After the collision, the real state of (AB)
consists precisely of the real state of A and the real state of B,
which two states have nothing to do with one another. \emph{The real
state of B thus cannot depend upon the kind of measurement I carry
out on A.} ('Separation hypothesis' from above.) But then for the
same state of B there are two (in general arbitrarily many) equally
justified $\psi_{B}$, which contradicts the hypothesis of a
one-to-one or complete description of the real states.''
\end{quote}

Einstein is clearly presuming separability with his assertion that
``the real state of (AB) consists precisely of the real state of A
and the real state of B''.
He furthermore appeals to local causality when he asserts that ``The
real state of B thus cannot depend upon the kind of measurement I
carry out on A'', because he is ruling out the possibility of events
at A having causes in the space-like separated region B.

Now, although Einstein's conclusion is nominally to deny
$\psi$-completeness, he does so by showing that there can be many
quantum states associated with the same ontic state, ``for the same
state of B there are two (in general arbitrarily many) equally
justified $\psi_{B}$''. The proof need not have taken this form. An
alternative approach would have been to try to deny
$\psi$-completeness by showing that there are many ontic states
associated with the same quantum state. \ For our purposes, this
distinction is critical because what Einstein has shown through his
argument is that a variation in $\psi$ need not correspond to a
variation in the ontic state. \ Recalling Definition
\ref{DEF:psi_ontic}, we see that Einstein has established the
failure of $\psi $-onticness!\ His 1935 incompleteness argument
rules out $\psi$-onticness \emph{en route} to ruling out
$\psi$-completeness.

The structure of his argument, in our terminology, is:
\begin{equation}
\text{L}\wedge\text{QSTAT}\rightarrow\lnot\text{(}\psi\text{-ontic)}%
\rightarrow\psi\text{-incomplete.}%
\end{equation}
But the second implication is actually a weakening of the
conclusion, because among the $\psi$-incomplete models are some
which are $\psi$-ontic (those we have called $\psi$-supplemented)
and the argument is strong enough to rule these out.

Einstein would have done better, therefore, to characterize his
argument as,
\[
\text{L}\wedge\text{QSTAT}\rightarrow\lnot \text{(}\psi\text{-ontic)},\\
\]

which is our Theorem \ref{THRM:psi_ont_nonlocal}.

\section{Historical Implications
\label{SEC:historical_implications}}

\subsection{A puzzle}
What can we gain from this retrospective assessment of Einstein's
incompleteness arguments? There is one long-standing puzzle that it
helps to solve: why did Einstein ever switch from the simple 1927
argument, which involves only a single measurement, to the 1935
argument, which involves two?

The move he made in 1935 to the two measurement argument described
in Sec.~\ref{SEC:historical_1935} proved to be a permanent one. He
published the argument for the first time in 1936 \cite{EPhysik} and
from this point onwards, the 1935 argument proved the mainstay of
his assault on orthodox quantum theory, appearing in various
writings \cite{Edialectica,EtoB}, most notably his own
autobiographical notes \cite{Eautobiog}. In fact, there is evidence
to suggest that this argument was still on Einstein's mind as late
as 1954 \cite{sauer}.


Many commentators have noted that an EPR-style argument for
incompleteness can be made even if one imagines that only a single
measurement is performed \cite{Hardy1995,Redhead,Fine,Maudlin}. The
resulting argument is similar to Einstein's 1927 argument, although
it differs insofar as it appeals to a pair of systems rather than a
single particle and makes use of the EPR criterion for reality
rather than the assumption of $\psi$-completeness. Nonetheless, the
point being made by these authors is the same as the one we have
just noted: having multiple possible choices of measurement is not
required to reach the conclusion of incompleteness from the
assumption of locality. Furthermore, the extra complication actually
detracts from the argument (whether it follows the reasoning of the
EPR paper or Einstein's correspondence with Schr\"{o}dinger),
because it introduces counterfactuals and modal logic into the game,
and this is precisely where most critics, including Bohr
\cite{Bohr_replytoEPR}, have focussed their attention. The single
measurement versions of the argument are, of course, completely
immune to such criticisms.

One explanation that has been offered for Einstein's move to two
measurements is that one can thereby land a harder blow on the
proponent of the orthodox approach by also defeating the uncertainty
principle in the course of the argument. Maudlin refers to this
``extra twist of the knife'' as ``an unnecessary bit of
grandstanding (probably due to Podolsky)''\cite{Maudlin}. Although
this may be an accurate assessment of what is going on in the EPR
paper, it does not explain Einstein's post-1935 conversion to the
two-measurement form of the argument. Indeed, Einstein explicitly
de-emphasizes the uncertainty principle in his own writings. For
instance, in his 1935 letter to Schr\"{o}dinger, he remarks: ``I
couldn't care less\footnote{``ist mir \textit{wurst}'' (emphasis in
original).} whether $\psi_{B}$ and $\psi_{\underline{B}}$ can be
understood as eigenfunctions of observables $B$,
$\underline{B}$''."\cite{EtoS1935}


Another explanation worth considering concerns the experimental
significance of the two gedankenexperiments. Although Einstein's
incompleteness arguments imply a dilemma between $\psi$-completeness
and locality, a sceptic who conceded the validity of the argument
could still evade the dilemma by choosing to reject some part of
quantum mechanics, specifically, those aspects that were required to
reach Einstein's conclusion. To eliminate this possibility, one
would have to provide experimental evidence in favor of these
aspects. From this perspective, there is a significant difference
between the 1927 and 1935 gedankenexperiments. In the case of the
former, the measurement statistics to which Einstein appeals
(perfect anti-correlation of measurements of local particle number)
can also be obtained from the mixed state
$\tfrac{1}{2}(|A\rangle\langle A|+|B\rangle\langle B|)$ rather than
the pure state $(1/\sqrt{2})(|A\rangle+|B\rangle)$. It follows that
the sceptic could avoid the dilemma by positing that such coherence
was illusory. To convince the sceptic, further experimental data --
for instance, a demonstration of coherence via interference -- would
be required. On the other hand, the measurement statistics of the
1935 gedankenexperiment cannot, in general, be explained under the
sceptic's hypothesis (which in this case amounts to positing a
separable mixed state). Indeed, \textit{any} hypothesis that takes
system $B$ to be in a mixture of pure quantum states (that are
unaffected by events at $A$) can be ruled out by the 1935 set-up
because the latter allows one to make predictions about the outcomes
of incompatible measurements on $B$ that are in violation of the
uncertainty principle. This has been demonstrated by Reid in the
context of the EPR scenario \cite{Reid} and by Wiseman \textit{et
al.}\cite{WJD07} more generally. Although Wiseman has argued that
this provides a reason for favoring the 1935 over the 1927 version
of Einstein's incompleteness argument \cite{Wiseman06}, he does not
suggest that it was Einstein's reason. Indeed, this is unlikely to
have been the case. Certainly, we are not aware of anything in
Einstein's writings that would suggest so.\footnote{Although
Schr\"{o}dinger had some doubts about the validity of quantum
theory, these concerned whether experiments would confirm the
existence of the steering phenomenon (``I am not satisfied about
there being enough experimental evidence for
that.''\cite{schroed_steer_quote}). This sentiment was a
\textit{reaction} to the 1935 form of Einstein's argument and so
could not have motivated it. It is unlikely that anyone would have
been sceptical of the spatial coherence assumed in Einstein's 1927
argument.}

\subsection{A possible explanation}
 Our analysis of Einstein's incompleteness arguments
suggests a very different explanation. In
Sec.~\ref{SEC:historical_1935}, we demonstrated that the 1935
argument is able to prove that both $\psi $-complete \textit{and}
$\psi$-supplemented models are incompatible with a locality
assumption, leaving $\psi$-epistemic models as the only approach
holding any hope of preserving locality. In contrast, the 1927
argument cannot achieve this stronger conclusion, as was noted in
Sec.~\ref{SEC:historical_1927}. (This also follows from the fact
that the deBroglie-Bohm theory constitutes a $\psi$-supplemented
model which provides a local explanation of the 1927 thought
experiment.) One can therefore understand Einstein's otherwise
baffling abandonment of his 1927 incompleteness argument in favor of
the more complicated 1935 one by supposing that he sought to
advocate a particular kind of ontological model, namely, a
$\psi$-epistemic one. This interpretation of events is bolstered by
the fact that Einstein often followed his discussions of the
incompleteness argument with an endorsement of the epistemic view of
quantum states. We turn to the evidence of his papers and
correspondence.

In addition to his conviction that ``[...] the description afforded
by quantum-mechanics is to be viewed [...] as an incomplete and
indirect description of reality, that will again be replaced later
by a complete and direct description.'' \cite{Edialectica}, Einstein
specifically advocated that
\begin{quote}
\textquotedblleft\lbrack t]he $\psi$-function is to be understood as
the description not of a single system but of an ensemble of
systems.\textquotedblright\ \cite{schilpp},
\end{quote}
and that the meaning of the quantum state was ``similar to that of
the density function in classical statistical
mechanics.''\cite{EtoBreit}

It is not immediately obvious that this is equivalent to an
epistemic interpretation of the quantum state. We argue for this
equivalence on the grounds that the ensembles Einstein mentions are
simply a manner of grounding talk about the probabilities that
characterize an observer's knowledge. In other words, the only
difference between ``ensemble talk'' and ``epistemic talk'' is that
in the former, probabilities are understood as relative frequencies
in an ensemble of systems, while in the latter, they are understood
as characterizations of the incomplete knowledge that an observer
has of a single system when she knows the ensemble from which it was
drawn. Ultimately, then, the only difference we can discern between
the ensemble view and the epistemic view concerns how one speaks
about probabilities, and although one can debate the merits of
different conceptions of probability, we do not feel that the
distinction is significant in this context, nor is there any
indication of Einstein having thought so.

Indeed, in a 1937 letter to Ernst Cassirer, Einstein seems to use
the two manners of characterizing his view \emph{interchangeably} as
he spells out what conclusion should be drawn from his 1935
incompleteness argument \cite{EtoCassirer},

\begin{quote}
``[...] this entire difficulty disappears if one relates $\psi_{2}$
not to an individual system but, in Born's sense, to a certain
state-ensemble of material points $2$. Then, however, it is clear
that $\psi_{2}$ does not describe the totality of what ``really''
pertains to the partial system $2$, rather only what we know about
it in this particular case.''
\end{quote}

\strut Einstein's endorsement of an epistemic understanding of the
quantum state is also explicit elsewhere in his personal
correspondence (of which relevant extracts have been conveniently
collected together in essays by Fine and Howard
\cite{Fine:Eincontext,Howard_einst_short,Howard_einst_long}). For
instance, in a 1945 letter to Epstein, after providing an
incompleteness argument containing all the features of the one used
in 1935, Einstein concludes that \cite{EtoEpstein},

\begin{quote}
``Naturally one cannot do justice to [the argument] by means of a
wave function. Thus I incline to the opinion that the wave function
does not (completely) describe what is real, but only a to us
empirically accessible maximal knowledge regarding that which really
exists [...] This is what I mean when I advance the view that
quantum mechanics gives an incomplete description of the real state
of affairs.''
\end{quote}

Perhaps the most explicit \strut example occurs in a 1948 reply to
Heitler, criticizing Heitler's notion that the observer plays an
important role in the process of wave-function collapse, and
advocating \cite{EtoHeitler},

\begin{quote}
\textquotedblleft that one conceives of the psi-function only as an
incomplete description of a real state of affairs, where the
incompleteness of the description is forced by the fact that
observation of the state is only able to grasp part of the real
factual situation. Then one can at least escape the singular
conception that observation (conceived as an act of consciousness)
influences the real physical state of things; the change in the
psi-function through observation then does not correspond
essentially to the change in a real matter of fact but rather to the
alteration in \textit{our knowledge} of this matter of
fact.\textquotedblright (emphasis in original)
\end{quote}

The result, implicit in Einstein's 1935 argument, that the only
realistic interpretation of quantum states that could possibly be
local are $\psi$-epistemic, is of course superseded by Bell's
theorem \cite{Bell_locality}. The latter famously demonstrates that
\textit{any} theory providing an adequate description of nature must
violate locality (as emphasized in
Refs.~\cite{norsen_blrealism,norsen_epr}). We do not dispute this.
The point we wish to make is simply that the `big guns' of Bell's
theorem are \textit{only needed to deal with $\psi $-epistemic
models}. Any $\psi$-ontic model can be seen to be non-local by an
argument that appeared in print as far back as 1936.

Therefore, in the 28 years between the publication of Einstein's
1935 incompleteness argument (in 1936) and the publication of Bell's
theorem (in 1964), only $\psi$-epistemic ontological models were
actually viable to those who were daring enough to defy convention
and seek an interpretation that preserves locality. Why is it then
that during the pre-Bell era, there was not a greater recognition
among such researchers of the apparent promise of $\psi$-epistemic
approaches vis-a-vis locality?

It seems likely to us that the distinction between
$\psi$-supplemented and $\psi$-epistemic hidden variable models was
simply not sufficiently clear. One searches in vain for any
semblance of a distinction in Einstein's description of the
alternative to the orthodox $\psi$-complete view during the general
discussion at the 1927 Solvay conference. But nothing in what we
have said would lead one to expect that Einstein had clearly
understood the distinction as early as 1927. What \textit{is}
surprising is that, after 1935, Einstein seems to voice his support
for an epistemic view of $\psi$ in his papers and correspondence,
and yet never bothers to articulate, nor explicitly denounce, the
other way in which his bijective notion of completeness
($\psi$-completeness) could fail, namely, by $\psi$ being ontic but
supplemented with additional variables.

By characterizing his 1935 argument as one that merely established
the \textit{incompleteness} of quantum theory on the assumption of
locality, Einstein did it a great disservice. For in isolation, a
call for the \textit{completion} of quantum theory would naturally
have led many to pursue hidden variable theories that interpreted
the fundamental mathematical object of the theory, the wave
function, in the same manner in which the fundamental object of
other physical theories were customarily treated -- as ontic. But
such a strategy was known by Einstein to be unable to preserve
locality.
Thus it is likely that the force of Einstein's 1935 argument from
locality to the epistemic interpretation of $\psi$ was not felt
simply because the argument was not sufficiently well articulated.

A proper assessment of the plausibility of these historical
possibilities would require a careful reexamination of Einstein's
papers and correspondence with the distinction between $\psi$-ontic
and $\psi$-epistemic ontological models in mind. We hope that such a
reassessment might yield further insight into the history of
incompleteness and nonlocality arguments.

\section{The future of $\psi$-epistemic models
\label{SEC:discussion}}

Bell's theorem shows that the preservation of locality is not a
motivation for a $\psi$-epistemic ontological model, because it
cannot be maintained. However, it does not provide any reason for
preferring a $\psi$-ontic approach over one that is
$\psi$-epistemic; it is neutral on this front. Moreover, there are
many new motivations (unrelated to locality) that can now be
provided in favor of $\psi$-epistemic models. For instance, it is
shown in Refs.~\cite{toy_theory,BRSLiouville} that
information-theoretic phenomena such as teleportation, no-cloning,
the impossibility of discriminating non-orthogonal states, the
information-disturbance trade-off, aspects of entanglement theory,
and many others, are found to be derivable within toy theories that
presume hidden variables and wherein the analogue of $\psi$ is a
state of incomplete knowledge. This interpretation of $\psi$ is
further supported by a great deal of foundational work that does not
presuppose hidden variables
\cite{Emerson,Fuchs,FuchsJmodopt,Ballentine70,Ballentine94,Peierls,Leiferarxiv,Leiferpra,CFS02,CFS02arxiv,CFS06}.
$\psi$-epistemic ontological models are therefore deserving of more
attention than they have received to date.

However, it remains unclear to what extent a $\psi$-epistemic
ontological model of quantum theory is even possible. Recall that
the Kochen-Specker model discussed in Sec.~\ref{SEC:example_ksmodel}
secured such an interpretation for pure states and projective
measurements in a two-dimensional Hilbert space. But can one be
found in more general cases? \footnote{Hardy was perhaps the first
to lay down this challenge explicitly \cite{Hardyprivate}.}

We here need to dispense with a possible confusion that might arise.
In the same paper wherein they presented their 2d model, Kochen and
Specker proceed to prove a no-go theorem for certain kinds of
ontological models seeking to reproduce the predictions of quantum
mechanics in 3d Hilbert spaces. \ One might therefore be led to the
impression that Kochen and Specker rule out $\psi$-epistemic models
for 3d Hilbert spaces. \ This is not the case, however, as we now
clarify.

As soon as one moves to projective measurements in a Hilbert space
of dimension greater than two, it is possible to define a
distinction between contextual and noncontextual ontological models
\cite{Spekkens_con}. \ It was famously shown by Bell
\cite{Bell_probhv} and independently by Kochen and Specker \cite{Ks}
that noncontextual ontological models cannot reproduce the
predictions of quantum theory for Hilbert space dimension 3 or
greater. \ Furthermore, the notion of noncontextuality can be
extended from projective measurements to nonprojective measurements,
preparations, and transformations \cite{Spekkens_con}. \ In all
cases, one can demonstrate a negative verdict for noncontextual
models of quantum theory \cite{Spekkens_con}. \ Indeed, by moving
beyond projective measurements, one finds that noncontextual models
cannot even be constructed for a two-dimensional Hilbert space.\

But the dichotomy between contextual and noncontextual models is
independent of the dichotomy between $\psi$-ontic and
$\psi$-epistemic models. So, whereas the Bell-Kochen-Specker theorem
and variants thereof show the necessity of contextuality, these are
silent on the issue of whether one can find an ontological model
that is also $\psi$-epistemic. The ontological models of quantum
theory that we do have, such as deBroglie-Bohm, are contextual but
$\psi$-ontic. Bell \cite{Bell_probhv} even provides a very \emph{ad
hoc} example of a contextual hidden variable model (an extension of
the Bell-Mermin model of Sec.~\ref{SEC:example_bellmermin}) to prove
that such a model is possible. It too is $\psi$-ontic (although one
must have recourse to the definition appealing to fidelities
provided in footnote \ref{FNOTE:fidelity} to properly assess this
model) \cite{Barrettpc}.

Many features of deBroglie-Bohm theory have been found to be
generalizable to a broad class of ontological models. Nonlocality,
contextuality, and signalling outside of quantum equilibrium
\cite{valentini} are examples. Inspired by this pattern, Valentini
has wondered whether the pilot-wave (and hence ontic) nature of the
wave function in the deBroglie-Bohm approach might be unavoidable
\cite{Valentiniprivate}. On the other hand, it has been suggested by
Wiseman that there exists an unconventional reading of the
deBroglie-Bohm approach which is not $\psi$-ontic
\cite{Wisemanprivate}. A distinction is made between the quantum
state of the universe and the conditional quantum state of a
subsystem, defined in Ref.~\cite{DGZ_conditionalqstate}. The latter
is argued to be epistemic while the former is deemed to be nomic,
that is, law-like, following the lines of
Ref.~\cite{DurrGoldteinZanghi} (in which case it is presumably a
category mistake to try to characterize the universal wave function
as ontic or epistemic). We shall not provide a detailed analysis of
this claim here, but highlight it as an interesting possibility that
is deserving of further scrutiny. Nelson's approach to quantum
theory \cite{Nelsonbook} also purports to \textit{not} assume the
wave function to be part of the ontology of the theory
\cite{BacconNelson}. However, as pointed out by Wallstrom
\cite{Wallstrom}, the theory does not succeed in picking out all and
only those solutions of Schr\"{o}dinger's equation\footnote{It is
assumed that only continuous and single-valued wave functions are
valid, a fact that is disputed by Smolin \cite{Smolin_Nelson}.}.
Consequently, it also fails to provide a $\psi$-epistemic model of
quantum theory.

Recently, Barrett \cite{Barrettpc} has constructed a model that is
$\psi$-epistemic.  Although it only works for a countable set of
bases of the Hilbert space, it seems likely that this deficiency can
be eliminated, in which case it would be the first $\psi$-epistemic
model for a Hilbert space of arbitrary dimension. Unfortunately, the
model achieves the $\psi$-epistemic property in a very ad hoc
manner, by singling out a pair of non-orthogonal quantum states, and
demanding that the epistemic states associated with these have
non-zero overlap, while the quantum predictions are still
reproduced.  It consequently does not have the sorts of features,
outlined in Refs.~\cite{toy_theory,BRSLiouville}, that make the
$\psi$-epistemic approach compelling.  This suggests that the
interesting question is not simply whether a $\psi$-epistemic model
can be constructed, but whether one can be constructed with certain
additional properties, such as the property that the classical
fidelity between epistemic states associated with a given pair of
quantum states is invariant under all unitary transformations of the
latter.\footnote{The Kochen-Specker model discussed in
Sec.~\ref{SEC:example_ksmodel} has this feature.}

Rudolph has devised a $\psi$-epistemic contextual ontological model
that is quantitatively close to the predictions of quantum theory
for projective measurements in three-dimensional Hilbert spaces and
also has the desired symmetry property \cite{tr_model}. This model
does not, however, reproduce the quantum predictions exactly.

It is possible that a $\psi$-epistemic model with the desired
symmetry property does not exist.  However, a no-go theorem always
presumes some theoretical framework. In Sec.~\ref{SEC:om_intro} of
the present paper, we have cast ontological models in an operational
framework, wherein systems are considered in isolation and the
experimental procedures are treated as external interventions. Such
a framework may not be able to do justice to all interpretations
that have some claim to being judged realist. For instance, in
deBroglie-Bohm, a system is not separable from the experimental
apparatus and consequently it is unclear whether one misrepresents
the interpretation by casting it in our current framework (an
extension of the formalism used here is, however, to be developed in
Ref.~\cite{deficiency}). Ontological models that are fundamentally
relational might also fail to be captured by the framework described
here. Nonetheless, something would undeniably be learned if one
could prove the impossibility of a $\psi$-epistemic model with the
desired symmetry properties within an operational framework of this
sort.

\section{Acknowledgements}

We would like to acknowledge Jonathan Barrett, Travis Norsen, and
Howard Wiseman for discussions and comments, and Don Howard and
Arthur Fine for their Einstein scholarship, without which the
present work would not have been possible. We are also grateful to
Terry Rudolph for numerous discussions on this work and for having
supported the progressive rock movement by refusing to adopt a
reasonable haircut. RWS acknowledges support from the Royal Society.
NH is supported by Imperial College London and the occasional
air-guitar recital.

\end{document}